\documentclass[11pt,a4paper]{article}

\usepackage[utf8]{inputenc}
\usepackage{amsmath}
\usepackage{amsfonts}
\usepackage{amssymb}
\usepackage{amsthm}

\title{Polynomial Space Randomness in Analysis}
\author{Xiang Huang and D. M. Stull\\
		\small{Iowa State University}\\
		\small{Department of Computer Science}\\
		\small{Ames, IA 50011 USA}\\
		\small{$\{$huangx, dstull$\}$@iastate.edu}}
\date{}

\newtheorem{theorem}{Theorem}
\newtheorem{lemma}{Lemma}

\newtheorem*{maintheorem}{Main Theorem}
\newtheorem*{lebesgue}{Lebesgue Differentiation Theorem}
\newtheorem*{cheby}{Chebyshev's Inequality}
\newtheorem*{hardy}{Hardy/Littlewood Inequality}
\label{rutetheorem}
\newtheorem*{BMNtheorem}{Theorem}
\theoremstyle{definition}
\newtheorem*{definition}{Definition}
\begin{document}

\maketitle
\begin{abstract}
We study the interaction between polynomial space randomness and a fundamental result of analysis, the Lebesgue differentiation theorem. We generalize Ko's framework for polynomial space computability in $\mathbb{R}^n$ to define \textit{weakly pspace-random} points, a new variant of polynomial space randomness. We show that the Lebesgue differentiation theorem characterizes weakly pspace random points. That is, a point $x$ is weakly pspace random if and only if the Lebesgue differentiation theorem holds for a point $x$ for every pspace $L_1$-computable function.
\end{abstract}

\section{Introduction}
The theory of computing allows for a meaningful definition of an individual point of Euclidean space being ``random". Classically, such a notion would seem paradoxical, as any singleton set (indeed, any countably infinite set) has measure zero. Martin-L\"of used computability to give the first mathematically robust definition of a point being random \cite{Lof}. Since Martin-L\"of's original definition, many notions of randomness have been introduced. In addition to Martin-L\"of randomness, two of the most prominent variants are Schnorr randomness and computable randomness \cite{DH}. By developing a theory of resource-bounded measure, Lutz initiated the study of resource-bounded randomness \cite{Lutz92, Lutz98}. This allowed for research in algorithmic randomness to extend to resource-bounded computation \cite{Wang}.

Recently, research in algorithmic randomness has used computable analysis to study the connection between randomness and classical analysis \cite{BDHMS, FGMN, FKN, GT, Miyabe, MNZ, Rute}. With the rise of measure theory, many fundamental theorems of analysis have been ``almost everywhere" results. Theorems of this type state that a certain property holds for almost every point; i.e., the set of points that does not satisfy the property is of measure zero. However, almost everywhere theorems typically give no information about which points satisfy the stated property. By adding computability restrictions, tools from algorithmic randomness are able to strengthen a theorem from a property simply holding almost everywhere, to one that holds for all random points. For example, an important classical result of analysis is Lebesgue's theorem on nondecreasing functions. Lebesgue showed that every nondecreasing continuous function $f:[0, 1] \rightarrow \mathbb{R}$ is differentiable almost everywhere. Brattka, Miller and Nies characterized computable randomness using Lebesgue's theorem by proving the following result \cite{BMN}.
\begin{BMNtheorem}
Let $z \in [0, 1]$. Then $z$ is computably random if and only if $f^\prime(z)$ exists for every nondecreasing computable function $f:[0, 1] \rightarrow \mathbb{R}$.
\end{BMNtheorem}
This paper concerns a related theorem, also due to Lebesgue \cite{Lebesgue}.
\begin{lebesgue}
For each $f \in L_1([0, 1]^n)$,
\begin{equation*}
f(x) = \lim_{Q \rightarrow x} \frac{\int_Q f d\mu}{\mu(Q)}
\end{equation*}
for almost every $x \in [0, 1]^n$. The limit is taken over all open cubes $Q$ containing $x$ as the diameter of $Q$ tends to 0.
\end{lebesgue}
Pathak first studied the Lebesgue differentiation theorem in the context of Martin-L\"of randomness \cite{Pathak}. Under the assumption that the function is $L_1$-computable, Pathak showed that the Lebesgue differentiation theorem holds for every Martin-L\"of random point. Subsequently, Pathak, Rojas and Simpson improved this theorem \cite{PRS}. They showed that the Lebesgue differentiation theorem holds at a point $z$ for every $L_1$ computable function if and only if $z$ is Schnorr random \cite{PRS}. Independently, and using very different techniques, Rute also showed that the Lebesgue differentiation theorem holds for Schnorr random points \cite{Rute}.

This paper concerns the connection between \textit{resource-bounded} randomness and analysis. While there has been work on this interaction \cite{BJL, LL, Nies}, resource-bounded randomness in analysis is still poorly understood. Recently, Nies extended the result of Brattka, Miller and Nies to the polynomial time domain \cite{Nies}. Specifically, Nies characterized polynomial time randomness using the differentiability of nondecreasing polynomial time computable functions. In this paper, we extend this research of the Lebesgue differentiation theorem to the context of resource-bounded randomness. We show that the Lebesgue differentiation theorem characterizes \textit{weakly polynomial space randomness}. We note that the polynomial space variant of Nies' result implies our result in one dimension. However, as in classical analysis, the proof for arbitrary dimension requires significantly different tools.

In order to work with resource bounded computability over continuous domains, we use the framework for polynomial space computability in $\mathbb{R}^n$ developed by Ko \cite{Ko}. Using generalizations of Ko's \textit{polynomial space approximable} sets, we define weakly polynomial space randomness, a new variant of polynomial space randomness. We prove that Lutz's notion of polynomial space randomness implies weakly polynomial randomness. Weakly polynomial space randomness uses open covers, similar to Martin-L\"of's original definition, unlike the martingale definitions commonly used in resource-bounded randomness. The use of open covers lends itself better to adapting many theorems of classical analysis. We believe that the notion of weakly polynomial space randomness will be useful in further investigations of resource-bounded randomness in analysis. 

Using this definition of randomness, we extend the result of Pathak, et al, and Rute to polynomial space randomness. Specifically, we prove that a point $x$ is weakly polynomial space random if and only if the Lebesgue differentiation theorem holds at $x$ for every polynomial space $L_1$-computable function. Structurally, the proof of this theorem largely follows that of Pathak, et al. However, the restriction to polynomial space forces significant changes to the internal methods. To prove the converse of our theorem, we introduce \textit{dyadic tree decompositions}. Intuitively, a dyadic tree decomposition partitions an open cover randomness test into a tree structure. This allows for the construction of a polynomial space $L_1$-computable function so that the Lebesgue differentiation theorem fails for any point covered by the test. We believe that dyadic tree decompositions will be useful in further research.

\section{Preliminaries}
Throughout the paper, $\mu$ will always denote the Lebesgue measure on $\mathbb{R}^n$. We denote the set of all Lebesgue integrable functions $f:[0, 1]^n \rightarrow \mathbb{R}$ by $L_1([0, 1]^n)$. A \textit{dyadic rational number d} is a rational number that has a finite binary expansion; that is $d = \frac{m}{2^r}$ for some integers $m$, $r$ with $r \geq 0$. We denote the set of all dyadic rational numbers by $\mathbf{D}$. We denote the set of all dyadic rationals $d$ of precision $r$ by $\mathbf{D}_r$. Formally,
\begin{equation*}
\mathbf{D}_r = \{\frac{m}{2^r} \, | \, m \in \mathbb{Z}\}.
\end{equation*}
We denote the set of dyadic rationals in the interval $[0, 1]$ by $\mathbf{D}[0, 1]$. We denote the set of dyadic rationals of precision $r$ in the interval $[0, 1]$ by $\mathbf{D}_r[0, 1]$.
An \textit{open dyadic cube of precision r} is a subset $Q \subseteq \mathbb{R}^n$ such that 
\begin{equation*}
Q = (\frac{a_1}{2^r}, \frac{a_1 + 1}{2^r}) \times \ldots \times (\frac{a_n}{2^r}, \frac{a_n + 1}{2^r}),
\end{equation*}
where $a_i \in \mathbb{Z}$, and $r \in \mathbb{N}$. We say that the points $\{\frac{a_1}{2^r}$, $\frac{a_1 + 1}{2^r}, \ldots \frac{a_n}{2^r}$, $\frac{a_n + 1}{2^r}\}$ are the \textit{endpoints of $Q$}. In the same manner, we define closed dyadic cubes, and half-open dyadic cubes. We denote the set of all open dyadic cubes of precision $r$ by
\begin{equation*}
\mathbf{B}_r = \{Q \, | \, Q \mbox{ is an open dyadic cube of precision }r\}.
\end{equation*}
For an open set $Q \subseteq \mathbb{R}^n$ and $t \in \mathbb{R}^n$, define the translation of $Q$ by $t$ to be the set 
\begin{equation*}
t + Q = \{t + x \, | \, x \in Q\}.
\end{equation*}
\subsection{Resource-Bounded Randomness in $\mathbb{R}^n$}
Lutz and Lutz recently adapted resource-bounded randomness to arbitrary dimension \cite{LL}. In this section, we review their definition of polynomial space randomness in $\mathbb{R}^n$. 

Let $r \in \mathbb{N}$, $\mathbf{u} = (u_1,\ldots, u_n) \in \mathbb{Z}^n$. Define the \textit{r-dyadic cube at $\mathbf{u}$} to be the half-open dyadic cube of precision r,
\begin{equation*}
Q_r(\mathbf{u}) = [u_1 \cdot 2^{-r}, (u_1 + 1) \cdot 2^{-r}) \times \ldots \times [u_n \cdot 2^{-r}, (u_n + 1) \cdot 2^{-r}).
\end{equation*}
Define the family 
\begin{equation*}
\mathcal{Q}_r = \{Q_r(\mathbf{u}) \; \vert \; \mathbf{u} \in \{0,\ldots, 2^r - 1\}^n \}.
\end{equation*}
So then $\mathcal{Q}_r$ is a partition of the unit cube $[0, 1)^n$. The family
\begin{equation*}
\mathcal{Q} = \mathop{\bigcup}_{r = 0}^{\infty} \mathcal{Q}_r,
\end{equation*}
is the set of all half-open dyadic cubes in $[0, 1)^n$.

A \textit{martingale} on $[0, 1)^n$ is a function $d: \mathcal{Q} \rightarrow [0, \infty)$ satisfying
\begin{equation}
d(Q_r(\mathbf{u})) = 2^{-n} \mathop{\sum}_{\mathbf{a} \in \{0, 1\}^n} d(Q_{r + 1}(2\mathbf{u} + \mathbf{a})),
\end{equation}
for all $Q_r(\mathbf{u}) \in \mathcal{Q}$. We may think of a martingale $d$ as a strategy for placing successive bets on which cube contains $x$. After $r$ bets have been placed, the bettor's capital is 
\begin{equation*}
d^{(r)}(\mathbf{x}) = d(Q_r(\mathbf{u})),
\end{equation*}  
where $\mathbf{u}$ is the unique element of $\{0,\ldots, 2^r - 1\}^n$ such that $\mathbf{x} \in Q_r(\mathbf{u})$. A martingale $d$ \textit{succeeds} at a point $\mathbf{x} \in [0, 1)^n$ if 
\begin{equation*}
\limsup\limits_{r \rightarrow \infty} d^{(r)}(\mathbf{x}) = \infty.
\end{equation*}
Let
\begin{equation*}
J = \{(r, \mathbf{u}) \in \mathbb{N} \times \mathbf{Z}^n \, | \, \mathbf{u} \in \{0,\ldots, 2^r - 1\}^n \}.
\end{equation*}
We say that a martingale $d: \mathcal{Q} \rightarrow [0, \infty)$ is \textit{computable} if there is a computable function $\hat{d}: \mathbb{N} \times J \rightarrow \mathbb{Q} \cap [0, \infty)$ such that for all $(s, r, \mathbf{u}) \in \mathbb{N} \times J$,
\begin{equation}\label{compeq}
|\hat{d}(s, r, \mathbf{u}) - d(Q_r(\mathbf{u}))| \leq 2^{-s}.
\end{equation}
A martingale $d: \mathcal{Q} \rightarrow [0, \infty)$ is \textit{p-computable} (resp. \textit{pspace-computable}) if there is a function $\hat{d}: \mathbb{N} \times J \rightarrow \mathbb{Q} \cap [0, \infty)$ that satisfies \ref{compeq} and is computable in $(s + r)^{O(1)}$ time (resp. space). A point $x \in \mathbb{R}^n$ is \textit{p-random} (resp. \textit{pspace-random}) if no p-computable (resp. pspace-computable) martingale succeeds at $x$.

\subsection{Polynomial Space Computability in $\mathbb{R}^n$}
In this section, we review Ko's framework for complexity theory in $\mathbb{R}^n$ \cite{Ko}. For the remainder of the paper, we include the write tape when considering polynomial space bounds of Turing machines.

We first introduce the polynomial space $L_1$-computable functions, the class of functions we will be using in the proof of the Lebesgue differentiation theorem. This definition is equivalent to Ko's notion of \textit{pspace approximable functions}. It is a direct analog of the $L_1$-computable functions used in computable analysis.

A function $f: [0, 1]^n \rightarrow \mathbb{R}$ is a \textit{simple step function} if $f$ is a step function such that
\begin{enumerate}
\item $f(x) \in \mathbf{D}$ for all $x \in [0, 1]^n$ and
\item there exists a finite number of (disjoint) dyadic boxes $Q_1,\ldots, Q_k$ and dyadic rationals $d_1,\ldots, d_k$ such that $f(x) = \sum\limits_{i = 1}^{k}d_i\chi_{Q_i}(x)$, where $\chi_Q$ is the characteristic function of a set $Q$.
\end{enumerate}

A function $f \in L_1([0, 1]^n)$ is \textit{polynomial space $L_1$-computable} if there exists a sequence of simple step functions, $\{f_m\}_{m \in \mathbb{N}}$, and a polynomial $p$ such that for all $d \in \mathbf{D}^n$,
\begin{enumerate}
\item $f_m(x) = \sum\limits_{i = 1}^{k}d_i\chi_{Q_i}(x)$, such that the endpoints of each $Q_i$ are in $\mathbf{D}^n_{p(m)}$,
\item there is a polynomial space TM $M$ computing $f_m$ in the sense that \\
$M(0^m, d) = \begin{cases} f_m(d) &\mbox{if } d$ is not a breakpoint of $f_m \\ 
\# & \mbox{otherwise }\end{cases}$
\item $\|f - f_m\|_1 \leq 2^{-n}$ .
\end{enumerate}

Note that we may assume that the polynomial $p$ is increasing. We will frequently use the following nice property of polynomial space $L_1$-computable functions. If $f \in L_1([0, 1]^n)$ is approximated by sequence of simple step function $\{f_m\}$ at precision $p$, then for every $i > 0$, $f_i$ is a constant function on every $Q \in \mathbf{B}_{p(i)}$.

An infinite sequence $\{S_m\}_{m \in \mathbb{N}}$ of finite unions of open boxes is \textit{polynomial space computable} if there exists a polynomial space TM $M$ such that for all $m > 0$, and all $d \in \mathbf{D}^n$,
\begin{equation*}
M(0^m, d) = \begin{cases} 
				1 &\mbox{if } d \in S_m \\ 
				-1 &\mbox{if } d$ is a boundary point of $S_m \\
				0 &\mbox{otherwise} \end{cases}
\end{equation*}
A set $S \subseteq [0, 1]^n$ is \textit{polynomial space approximable} if $S$ is measurable and there exists a polynomial space computable sequence of sets $\{S_m\}_{m \in \mathbb{N}}$ such that, for every $m > 0$,
\begin{enumerate}
\item there is a polynomial $p$ such that all endpoints of $S_m$ are in $\mathbf{D}^n_{p(m)}$ and
\item $\mu(S \Delta S_m) \leq 2^{-m}$.
\end{enumerate}
Note that we may assume that the polynomial $p$ is increasing; that is $p(i) \leq p(i + 1)$, for all $i \in \mathbb{N}$.

\section{Uniformly Approximable Sequences}

We now generalize Ko's definition of approximable sets to approximable \textit{arrays} of sets. We follow Ko in first defining computability, then leveraging this to define approximability.
\begin{definition}
An infinite array $\{S_m^k\}_{k, m \in \mathbb{N}}$ of finite unions of open boxes is \textit{uniformly polynomial space computable} if there exists a polynomial space TM $M$ such that for all $k, m > 0$, and all $d \in \mathbf{D}^n$,
\begin{equation*}
M(0^m, 0^k, d) = \begin{cases} 
				1 &\mbox{if } d \in S_m^k \\ 
				-1 &\mbox{if } d$ is an boundary point of $S_m^k \\
				0 &\mbox{otherwise} \end{cases}
\end{equation*}
If $\{S_m^k\}$ is uniformly polynomial space computable and $M$ is a TM satisfying the definition, we say $M$ \textit{computes} $\{S^k_m\}$.
\end{definition}

\begin{definition}
A sequence of sets $\{U_m\}_{m \in \mathbb{N}}$ is \textit{uniformly polynomial space approximable} if there exists a uniformly polynomial space computable array of sets $\{S_m^k\}$ and a polynomial $p$ such that
\begin{enumerate}
\item all endpoints of $S_m^k$ are in $\mathbf{D}^n_{p(m + k)}$ and
\item $\mu(U_m \Delta S_m^k) \leq 2^{-k}$.
\end{enumerate}
If a polynomial $p$ and a uniformly polynomial space computable sequence $\{S^k_m\}$ satisfies (1) and (2), we say that $\{S^k_m\}_{k, m \in \mathbb{N}}$ \textit{approximates} $\{U_m\}$ \textit{at precision $p$}. Note that we may assume that the polynomial $p$ is increasing.
\end{definition}

We now show that we can construct \textit{uniformly} pspace computable sequences from pspace computable sequences. This lemma will be useful, as polynomial space computability is an easier property to verify than its uniform counterpart. 
\begin{lemma}\label{comptounif}
Let $\{T_i\}_{i \in \mathbb{N}}$ be a polynomial space computable sequence of sets, and $q_1$, $q_2$ be polynomials. For every $k$, $m > 0$, define the set $S^k_m$ by 
\begin{equation*}
S^k_m = \bigcup\limits_{i = q_1(m)}^{q_2(k)} T_i.
\end{equation*}
Then the array $\{S^k_m\}$ is uniformly polynomial space computable. 
\end{lemma}
\begin{proof}
It is clear that $S^k_m$ is a finite union of open boxes for each $k$ and $m > 0$. Let $M^\prime$ be the polynomial space TM computing $\{T_i\}$. For every $k$, $m  > 0$, and $d \in \mathbf{D}^n$, define the TM $M$ by

$M(0^m, 0^k, d) = \begin{cases} 
				1 &\mbox{if } M^\prime(0^i, d) = 1$ for any $q_1(m) \leq i \leq q_2(k) \\ 
				-1 &\mbox{else, if } M^\prime(0^i, 0^{2k + 2}, d) = -1$ for any $q_1(m) \leq i \leq q_2(k) \\
				0 &\mbox{otherwise} \end{cases}$.

Clearly, $M$ is computable in polynomial space. Hence, $\{S^k_m\}_{k, m \in \mathbb{N}}$ is uniformly polynomial space computable. 
\end{proof}

Similarly, we are able to construct uniformly pspace approximable sequences from other uniformly approximable sequences.
\begin{lemma}\label{approxtounif}
Let $q$ be a polynomial $j \in \mathbb{N}$, and $(V_i)$ be a uniformly polynomial space approximable sequence, such that $\mu(V_i) \leq 2^{-i + j}$. Define the sequence $\{U_m\}$ by
\begin{equation*}
U_m = \bigcup\limits_{i = q(m)}^\infty V_i.
\end{equation*} 
Then $\{U_m\}_{m \in \mathbb{N}}$ is a uniformly polynomial space approximable sequence.
\end{lemma}
\begin{proof}
Let $\{V_i\}$ be a uniformly approximable sequence, approximated by the uniformly pspace computable array $\{T^s_i\}$ at precision $p$. For each $k$, $m > 0$, define the set
\begin{equation*}
S^k_m = \bigcup\limits_{i = q(m)}^{k + j +1} T^{2k + 2}_i.
\end{equation*}
It is clear that $\{S^k_m\}_{k, m \in \mathbb{N}}$ is a array of finite unions of open boxes. Let $M^\prime$ be the polynomial space TM computing $\{T^s_i\}$. For every $k$, $m > 0$ and $d \in \mathbf{D}^n$, define the TM $M$ by

$M(0^m, 0^k, d) = \begin{cases} 
				1 &\mbox{if } M^\prime(0^i, 0^{2k + 2}, d) = 1$ for any $q(m) \leq i \leq k + j + 1 \\ 
				-1 &\mbox{else, if } M^\prime(0^i, 0^{2k + 2}, d) = -1$ for any $q(m) \leq i \leq k + j + 1 \\
				0 &\mbox{otherwise} \end{cases}$.
				
It is easy to see that $M$ is a polynomial space TM. Hence, $\{S^k_m\}_{k, m \in \mathbb{N}}$ is a uniformly pspace computable sequence. Recall that we are able to assume that the polynomial $p$ is increasing. Therefore, all endpoints of $S^k_m$ are in $\mathbf{D}^n_{p(3k + 3)}$. Finally, we have

\begin{align*}
\mu(U_m \Delta S^k_m) &= \mu(\bigcup\limits_{i = q(m)}^\infty V_i \, \Delta \bigcup\limits_{i = q(m)}^{k + j + 1} T^{2k + 2}_i) \\
&\leq \mu(\bigcup\limits_{i = q(m)}^{k + j + 1} V_i \, \Delta \bigcup\limits_{i = q(m)}^{k + j + 1} T^{2k + 2}_i) + \mu(\bigcup\limits_{i = k + j + 2}^{\infty} V_i) \\
&\leq \sum\limits_{i = q(m)}^{k + j + 1} \mu(V_i \Delta T^{2k + 2}_i) + \sum\limits_{i = k + j + 2}^{\infty} \mu(V_i) \\
&\leq \sum\limits_{i = q(m)}^{k + 1} 2^{-2k - 2} + \sum\limits_{i = k + j + 2}^{\infty} 2^{-i + j} \\
&\leq 2^{-k}.
\end{align*}
So then $\{S^k_m\}_{k, m \in \mathbb{N}}$ approximates $\{U_m\}_{m \in \mathbb{N}}$ at precision $p$, and therefore $\{U_m\}_{m \in \mathbb{N}}$ is a uniformly polynomial space approximable sequence.
\end{proof}

\section{Weakly Polynomial Space Randomness}
Using uniformly polynomial space approximable sequences, we give an open-cover definition of polynomial space randomness. This variant is intended to be similar to the open-cover definitions of the various computable randomness notions. However, the resource bounds force us to replace the typical enumerability requirements with approximability.
\begin{definition}
Let $a, b \in \mathbb{Z}$. An infinite sequence of open sets $\{U_m\}_{m \in \mathbb{N}} \subseteq [a, b]^n$ is a \textit{polynomial space $\mathcal{W}$-test (pspace $\mathcal{W}$-test)} if the following hold.
\begin{enumerate}
\item For every $m$, $\mu(U_m) \leq 2^{-m}$.
\item There is a uniformly pspace computable array $\{S^k_m\}$ approximating $\{U_m\}$ such that, for all $m$, 
\begin{equation*}
U_m \subseteq \liminf\limits_{k \rightarrow \infty} S^k_m, 
\end{equation*}
\end{enumerate}

A point $x$ \textit{passes} a polynomial space $\mathcal{W}$-test $\{U_m\}_{m \in \mathbb{N}}$ if $x \notin \bigcap\limits_{m = 1}^\infty U_m$. We say that $x$ is \textit{weakly pspace random} if $x$ passes every polynomial space $\mathcal{W}$-test. 
\end{definition}

The approximability of pspace $\mathcal{W}$-tests allows us to estimate the measure of the open covers in polynomial space.
\begin{lemma}\label{approxmeasure}
If $\{U_m\}_{m \in \mathbb{N}}$ is a pspace $\mathcal{W}$-test, then there exists a polynomial space TM $M$ such that for every $s$, $r$, $m \in \mathbb{N}$ and $\mathbf{u} \in \{0,\ldots, 2^r - 1\}^n$
\begin{equation*}
|M(0^s, 0^r, \mathbf{u}, 0^m) - \mu(U_m \cap Q_r(\mathbf{u}))| \leq 2^{-s}.
\end{equation*}
\end{lemma}
\begin{proof}
Let $p$ be a polynomial, and $\{U_m\}_{m \in \mathbb{N}}$ be a pspace $\mathcal{W}$-test, approximated by the uniformly pspace computable array $\{S^k_m\}$ at precision $p$. Let $M^\prime$ be the polynomial space TM computing $\{S^k_m\}_{k, m \in \mathbb{N}}$. For every $s$, $r$, $m \in \mathbb{N}$ and $\mathbf{u} \in \{0,\ldots, 2^r - 1\}^n$, define the TM $M$ by,
\begin{equation*}
M(0^s, 0^r, \mathbf{u}, 0^m) = \mu(S^s_m \cap Q_r(\mathbf{u})).
\end{equation*}
Then,
\begin{align*}
|M(0^s, 0^r, \mathbf{u}, 0^m) - \mu(U_m \cap Q_r(\mathbf{u}))| &= |\mu(S^s_m \cap Q_r(\mathbf{u})) - \mu(U_m \cap Q_r(\mathbf{u}))| \\
&\leq \mu((S^s_m \Delta U_m) \cap Q_r(\mathbf{u}))\\
&\leq 2^{-s}.
\end{align*}
It remains to be shown that $M$ is a polynomial space machine. To compute $\mu(S^s_m \cap Q_r(\mathbf{u}))$, $M$ enumerates over all dyadic cubes $Q$ of precision $p(s + m)$. For each $Q$, $M$ computes the center of $Q$, the dyadic rational $d_Q$ of precision $p(s + m) + 1$. If $M^\prime(0^m, 0^s, d_Q) = 1$, then $M$ adds $\mu(Q \cap Q_r(\mathbf{u}))$ to the current measure. After enumerating over all $Q \in \mathbf{B}_{p(s + m)}$, $M$ outputs the total measure. Hence, $M$ is a polynomial space machine, and the proof is complete.
\end{proof}

We are now able to relate weakly polynomial space randomness with Lutz's pspace randomness. The following lemma shows that pspace randomness implies weakly pspace randomness. 
\begin{lemma}\label{weaklytopspace}
Let $\{U_m\}_{m \in \mathbb{N}}$ be a polynomial space $\mathcal{W}$-test. Then there exists a pspace martingale $d$ succeeding on all points $x \in \bigcap\limits_{m = 1}^\infty U_m \, \bigcap \, [0, 1]^n$.
\end{lemma}
\begin{proof}
Let $\{U_m\}_{m \in \mathbb{N}}$ be a polynomial space $\mathcal{W}$-test. For each $m > 0$, define the function $d_m : \mathcal{Q} \rightarrow [0, \infty)$ by
\begin{equation*}
d_m(Q_r(\mathbf{u})) = \frac{1}{\mu(Q_r(\mathbf{u}))} \, \mu(U_m \cap Q_r(\mathbf{u})).
\end{equation*} 
We then have
\begin{align*}
2^{-n}\mathop{\sum}_{\mathbf{a} \in \{0, 1\}^n} d_m(Q_{r + 1}(2\mathbf{u} + \mathbf{a})) 
&= 2^{-n}\mathop{\sum}_{\mathbf{a} \in \{0, 1\}^n} \frac{1}{\mu(Q_{r + 1}(2\mathbf{u} + \mathbf{a}))} \, \mu(U_m \cap Q_{r+1}(2\mathbf{u} + \mathbf{a})) \\
&= 2^{rn} \mathop{\sum}_{\mathbf{a} \in \{0, 1\}^n}\mu(U_m \cap Q_{r+1}(2\mathbf{u} + \mathbf{a})) \\
&= 2^{rn} \mu(U_m \bigcap (\mathop{\bigcup}_{\mathbf{a} \in \{0, 1\}^n}Q_{r+1}(2\mathbf{u} + \mathbf{a}))) \\
&= 2^{rn} \mu(U_m \cap Q_r(\mathbf{u})) \\
&= \frac{1}{\mu(Q_r(\mathbf{u}))} \, \mu(U_m \cap Q_r(\mathbf{u})) \\
&= d_m(Q_r(\mathbf{u})),
\end{align*}
and so $d_m$ is a martingale. Define the function $d:\mathcal{Q} \rightarrow [0, \infty)$ by
\begin{equation*}
d(Q_r(\mathbf{u})) = \mathop{\sum}_{m = 1}^\infty d_m(Q_r(\mathbf{u})).
\end{equation*}
Then,  
\begin{align*}
d(Q_0(\mathbf{0})) &= \mathop{\sum}_{m = 1}^\infty d_m(Q_0(\mathbf{0})) \\
& \leq \mathop{\sum}_{m = 1}^\infty 2^{-m} \\
& \leq 1,
\end{align*}
and since each $d_m$ is a martingale, $d$ is a martingale. We now show that $d$ is a pspace martingale by constructing a polynomial space TM $M$ computing $\hat{d}$. By Lemma \ref{approxmeasure}, there exists a polynomial space TM $M^\prime$ such that
\begin{equation*}
|M^\prime(0^s, 0^r, \mathbf{u}, 0^m) - \mu(U_m \cap Q_r(\mathbf{u}))| \leq 2^{-s}.
\end{equation*} 
For every $s \in \mathbb{N}$ and $(r, \mathbf{u}) \in J$, define the TM $M$ by
\begin{align*}
M(0^s, 0^r, \mathbf{u}) &= \sum\limits_{m = 1}^{s + nr + 1} \frac{1}{\mu(Q_r(\mathbf{u}))}M^\prime(0^{s + nr + 2}, 0^r, \mathbf{u}, 0^m)\\
&= \sum\limits_{m = 1}^{s + nr + 1} 2^{nr}M^\prime(0^{s + nr + 2}, 0^r, \mathbf{u}, 0^m)
\end{align*}
Clearly, $M$ runs in polynomial space. Moreover,
\begin{align*}
|M(0^s, 0^r, \mathbf{u}) - d(Q_r(\mathbf{u}))| &= |M(0^s, 0^r, \mathbf{u}) - \mathop{\sum}_{m = 1}^\infty d_m(Q_r(\mathbf{u}))| \\
&\leq |M(0^s, 0^r, \mathbf{u}) - \mathop{\sum}_{m = 1}^{s + nr + 1}d_m(Q_r(\mathbf{u}))| + \mathop{\sum}_{m = s + nr + 2}^{\infty}d_m(Q_r(\mathbf{u})).
\end{align*}
By the definition of $M$,
\begin{align*}
|M(0^s, 0^r, \mathbf{u}) - \mathop{\sum}_{m = 1}^{s + nr + 1}d_m(Q_r(\mathbf{u}))| &= 2^{nr}|\mathop{\sum}_{m = 1}^{s + nr + 1} M^\prime(0^{s + nr + 2}, 0^r, \mathbf{u}, 0^m) - \mu(U_m \cap Q_r(\mathbf{u}))| \\
&\leq 2^{nr}\mathop{\sum}_{m = 1}^{s + nr + 1} 2^{-s - nr - 2} \\
&\leq \mathop{\sum}_{m = 1}^{s + nr + 1}2^{-s - 2} \\
&\leq 2^{-s - 1}.
\end{align*}
Combining the two inequalities, we have
\begin{align*}
|M(0^s, 0^r, \mathbf{u}) - d(Q_r(\mathbf{u}))| &\leq 2^{-s -1} + \mathop{\sum}_{m = s + nr + 2}^{\infty}d_m(Q_r(\mathbf{u})) \\
&\leq 2^{-s -1} + \mathop{\sum}_{m = s + nr + 2}^{\infty}2^{nr} \, 2^{-m} \\
&\leq 2^{-s -1} + 2^{nr}2^{-s - nr - 1} \\
&\leq 2^{-s}.
\end{align*}
Therefore, $d$ is a pspace martingale. 

Assume $x \in \bigcap\limits_{m = 1}^\infty U_m \, \bigcap \, [0, 1]^n$. Let $i > 0$. Then, since $U_i$ is an open set, there exists an $N$ such that for all $r \geq N$, $Q_r(\mathbf{u}) \subseteq U_i$, where $Q_r(\mathbf{u})$ is the unique dyadic cube containing $x$. Hence, for all $r \geq N$, $d_i(Q_r(\mathbf{u})) = 1$. Therefore,
\begin{equation*}
\lim\limits_{r \rightarrow \infty} d^{(r)}(x) = \infty,
\end{equation*}
and so $d$ succeeds on $x$.
\end{proof}

\section{Randomness and the Lebesgue Differentiation Theorem}
In this section we prove our main theorem, that the Lebesgue differentiation theorem characterizes weakly pspace-randomness. Recall the statement of Lebesgue's theorem.
\begin{lebesgue}
For each $f \in L_1([0, 1]^n)$,
\begin{equation*}
f(x) = \lim_{Q \rightarrow x} \frac{\int_Q f d\mu}{\mu(Q)}
\end{equation*}
for almost every $x \in [0, 1]^n$. The limit is taken over all open cubes $Q$ containing $x$ as the diameter of $Q$ tends to 0.
\end{lebesgue}

A point $x$ that satisfies the Lebesgue differentiation theorem is called a \textit{Lebesgue point}. We will prove the following theorem, 
\begin{maintheorem}
A point $x$ is weakly pspace-random if and only if for every polynomial space $L_1$-computable $f \in L_1([0, 1]^n)$, and every polynomial space computable sequence of simple functions $\{f_m\}_{m \in \mathbb{N}}$ approximating $f$,
\begin{equation}\label{maineq}
\lim_{m \rightarrow \infty} f_m(x) = \lim_{Q \rightarrow x} \frac{\int_Q f d\mu}{\mu(Q)}
\end{equation}
where the limit is taken over all cubes $Q$ containing $x$ as the diameter of $Q$ tends to 0.
\end{maintheorem}

We first make several remarks regarding the form of our main theorem. The use of polynomial space $L_1$-computability is not simply for the sake of generality. It is well-known that if a function is continuous, the Lebesgue differentiation theorem holds for \textit{every} point. Thus, to get a non-trivial randomness result, we must allow the function to be discontinuous. Our second remark concerns the limit of the approximating functions. In the statement of the classical theorem, the integral limit is equal to $f(x)$; whereas in our main theorem, it is equal to $\lim_{m \rightarrow \infty} f_m(x)$. This concession is necessary. For any point $x$, it is trivial to construct a polynomial space $L_1$-computable function $f$ such that
\begin{equation*}
f(x) \neq \lim_{Q \rightarrow x} \frac{\int_Q f d\mu}{\mu(Q)}.
\end{equation*}
Consider the function $f$ which is $0$ for all points, except at the given point $x$, $f(x) = 1$. Clearly, $f$ is polynomial space $L_1$-computable, but $x$ does not satisfy the Lebesgue differentiation theorem.

\subsection{Random points satisfy the Lebesgue differentiation theorem}
The outline of our proof roughly follows that of the classical proof of the Lebesgue differentiation theorem \cite{PRS, WZ}. However, the restriction to polynomial space computation significantly changes the internal methods. We first show that if a point $x \in [0, 1]^n$ is weakly pspace-random, then it must be contained in an open dyadic cube. This is a useful property of weakly pspace-random points that we take advantage of in later theorems.
\begin{lemma}\label{dyadic2random}
Let $x = (x_1,\ldots,x_n) \in [0, 1]^n$ be weakly pspace-random. Then, for every $i$, $x_i$ is not a dyadic rational.
\end{lemma}
\begin{proof}
Let $x = (x_1, \ldots, x_n) \in [0, 1]^n$ be weakly pspace-random. We show that $x_1$ cannot be a dyadic rational, the proof for the other components is similar. For every $i > 0$, define the set 
\begin{equation*}
S_i = \bigcup\limits_{d \in \mathbf{D}_i[0,1]} (d - 2^{-2i - 2}, d + 2^{-2i - 2}) \times (0, 1) \times \ldots \times (0, 1).
\end{equation*}
For every $m > 0$, define the set 
\begin{equation*}
U_m = \bigcup\limits_{i = m}^\infty S_i.
\end{equation*}
We now prove that the sequence $\{U_m\}_{m \in \mathbb{N}}$ is a pspace $\mathcal{W}$-test. It is clear that for every $m > 0$, $U_m$ is an open set. Let $m > 0$, then,
\begin{align*}
\mu(U_m) &= \mu(\bigcup\limits_{i = m}^\infty S_i) \\
&\leq \sum\limits_{i = m}^\infty \mu(S_i) \\
&\leq \sum\limits_{i = m}^\infty 2^{i} 2^{-2i - 1}  \\
&\leq 2^{-m}.
\end{align*}
It remains to be shown that $\{U_m\}_{m \in \mathbb{N}}$ is uniformly pspace approximable. For every $k$, $m > 0$, define the set
\begin{equation*}
T^k_m = \bigcup\limits_{i = m}^{k - 1} S_i.
\end{equation*}
It is easy to verify that $\{S_i\}$ is a polynomial space computable sequence. Hence, by Lemma \ref{comptounif}, $\{T^k_m\}$ is a uniformly polynomial space computable sequence. Finally, for every $k$, $m > 0$,
\begin{align*}
\mu(U_m \Delta T^k_m) &= \mu(\bigcup\limits_{i = k}^\infty S_i) \\
&\leq \sum\limits_{i = k}^\infty \mu(S_i) \\
&\leq 2^{-k},
\end{align*}
and so the sequence $\{U_m\}$ is uniformly polynomial space approximable. It is clear that for every $m$, and all $x \in U_m$, $x \in \liminf_k T^k_m$. Therefore, $\{U_m\}_{m \in \mathbb{N}}$ is a polynomial space $\mathcal{W}$-test. By assumption $x \notin \cap U_m$, therefore $x_1$ is not a dyadic rational. 

Using a similar argument we see that, for all $1 \leq i \leq n$, $x_i$ is not a dyadic rational.
\end{proof}

Let $f$ be a polynomial space $L_1$-computable function, approximated by the pspace computable sequence of simple step functions $\{f_m\}_{m \in \mathbb{N}}$. We now show that for every weakly pspace-random point $x$, the limit $\lim\limits_{m \rightarrow \infty}f_m(x)$ exists. We will need the following inequality due to Chebyshev. For every $f \in L_1([0, 1]^n)$ and $\epsilon > 0$, define the set
\begin{equation*}
S(f, \epsilon) = \{x \, | \, \, |f(x)| > \epsilon\}.
\end{equation*}
\begin{cheby}
Let $f \in L_1([0, 1]^n)$ and $\epsilon > 0$. Then $\mu(S(f, \epsilon)) \leq \frac{\|f\|_1}{\epsilon}$.
\end{cheby}
\begin{theorem}\label{limitrandom}
Let $f \in L_1([0, 1]^n)$ be polynomial space $L_1$ computable, approximated by the polynomial space computable sequence of simple step functions $\{f_m\}_{m \in \mathbb{N}}$. If $x$ is weakly pspace-random, the limit $\lim\limits_{m \rightarrow \infty}f_m(x)$ exists.
\end{theorem}
\begin{proof}
Let p be a polynomial and $f \in L_1([0, 1]^n)$ be polynomial space $L_1$ computable, approximated by the polynomial space computable sequence of simple step functions $\{f_m\}_{m \in \mathbb{N}}$ at precision $p$. Recall that we may assume that $p$ is increasing. For each $i \geq 1$, define the set
\begin{equation*}
S_i = (S(f_{2i-1} - f_{2i}, 2^{-i}) \cup S(f_{2i} - f_{2i + 1}, 2^{-i})) \cap (\bigcup\limits_{Q \in \mathbf{B}_{p(2i + 1)}} Q).
\end{equation*} 
We intersect with the open dyadic cubes of precision $p(2i + 1)$ to ensure that $S_i$ is an open set. For each $m \geq 1$ define the set 
\begin{equation*}
U_m = \bigcup\limits_{i = m + 4}^{\infty} S_i.
\end{equation*} 
We now prove that the sequence $\{U_m\}_{m \in \mathbb{N}}$ is a pspace $\mathcal{W}$-test. Using the properties of simple step functions, it is routine to verify that, for every $i > 0$, $S_i$ is the union of all open dyadic cubes $Q \in \mathbf{B}_{p(2i + 1)}$, such that either
\begin{align*}
|f_{2i - 1}(Q) - f_{2i}(Q)| &> 2^{-i}, \mbox{or} \\
|f_{2i}(Q) - f_{2i + 1}(Q)| &> 2^{-i}.
\end{align*} 
Therefore, for every $m > 0$, $U_m$ is an open set. By Chebyshev's inequality,
\begin{align*}
\mu(S_i) &\leq 2^i (\|f_{2i - 1} - f_{2i} \| + \|f_{2i} - f_{2i + 1} \| )\\
&\leq 2^i (2^{-2i + 2} + 2^{-2i + 1}) \\
&\leq 2^{-i + 3}.
\end{align*}
Using this upper bound on the measure of $S_i$ we obtain
\begin{align*}
\mu(U_m) &\leq \sum\limits_{i = m + 4}^\infty \mu(S_i) \\
&\leq \sum\limits_{i = m + 4}^\infty 2^{-i + 3} \\
&\leq 2^{-m}.
\end{align*}

It remains to be shown that the sequence $\{U_m\}_{m \in \mathbb{N}}$ is uniformly polynomial space approximable. For every $k$, $m > 0$, define the set
\begin{equation*}
T^k_m = \bigcup\limits_{i = m + 4}^{k + 3} S_i.
\end{equation*}
It is clear that $\{S_i\}$ is a polynomial space computable sequence. Hence, by Lemma \ref{comptounif}, $\{T^k_m\}$ is a uniformly pspace computable array. Finally, we have
\begin{align*}
\mu(U_m \Delta T^k_m) &= \mu(U_m \Delta (\bigcup\limits_{i = m + 4}^{k + 3} S_i)) \\
&\leq \mu((\bigcup\limits_{i = k + 4}^{\infty} S_i)) \\
&\leq \sum\limits_{i = k + 4}^\infty \mu(S_i) \\
&\leq \sum\limits_{i = k + 4}^\infty 2^{-i + 3} \\
&\leq  2^{-k}. 
\end{align*}
Finally, it is clear that, for every $m \in \mathbb{N}$ and all $x \in U_m$, $x \in \liminf_k T^k_m$. Hence, $\{U_m\}_{m \in \mathbb{N}}$ is a pspace $\mathcal{W}$-test.

Assume $x$ is weakly pspace-random. Then there exists an $N$ such that for all $m > N$, $x \notin U_m$, and therefore $x \notin S_i$, for all $i > N + 4$. By Lemma \ref{dyadic2random}, $x$ cannot have any dyadic rational components; i.e., $x \in Q$, for some $Q \in \mathbf{B}_{2i + 1}$. Hence, $|f_{2i - 1}(x) - f_{2i}(x)| \leq 2^{-i}$ and $|f_{2i}(x) - f_{2i+ 1}(x)| \leq 2^{-i}$. Let $j > 2N + 8$, then $|f_{j}(x) - f_{j + 1}(x)| \leq 2^{-\frac{j}{2}}$. Therefore, the limit $\lim\limits_{m \rightarrow \infty}f_m(x)$ exists.
\end{proof}

We now focus on the limit 
\begin{equation*}
\lim_{Q \rightarrow x} \frac{\int_Q f d\mu}{\mu(Q)}
\end{equation*}
on the right hand side of our main theorem (equation \ref{maineq}). The restriction to polynomial space computation creates difficulties in considering arbitrary open cubes. Intuitively, we overcome this obstacle through the use of translations of dyadic cubes, which are more amenable to polynomial space computation. Formally, for $t \in \{-\frac{1}{3}, 0, \frac{1}{3}\}^n$, define the set
\begin{equation*}
\mathbf{B}^t_r = \{I^t_r \, | \, I^t_r  = t + Q\mbox{, where }Q \in \mathbf{B}_r\}.
\end{equation*}
That is, $\mathbf{B}^t_r$ is the set of all translations of dyadic cubes of precision $r$ by points $t \in \{-\frac{1}{3}, 0, \frac{1}{3}\}^n$. For every $x \in [0, 1]^n$, let $I^t_r(x)$ denote the (unique) element of $\mathbf{B}^t_r$ containing $x$. The following theorem of Rute \cite{Rute}, using results due to Morayne and Solecki \cite{MS}, shows that it suffices to prove that the right hand limit of equation \ref{maineq} exists for these translations.
\begin{theorem}[\cite{Rute}]\label{rutestheorem}
Let $f \in L_1([0, 1]^n)$, and $x \in [0, 1]^n$. Then the following are equivalent,
\begin{enumerate}
\item the limit $\lim\limits_{Q \rightarrow x} \frac{\int_{Q} f d\mu}{\mu(Q)}$ exists, where the limit is taken over all cubes containing $x$, as the diameter goes to $0$
\item the limit $\lim\limits_{k \rightarrow \infty} \frac{\int_{I^t_k(x)} f d\mu}{\mu(I^t_k(x))}$ exists, for all $t \in \{-\frac{1}{3}, 0, \frac{1}{3}\}^n$.
\end{enumerate}
\end{theorem}
\bigskip
\noindent
We now show that the limit
\begin{equation*}
\lim\limits_{m \rightarrow \infty}\frac{\int_{I^t_r(x)} |f-f_{m}|d\mu}{\mu(I^t_r(x))}
\end{equation*}
exists, for every $t \in \{-\frac{1}{3}, 0, \frac{1}{3}\}^n$ and $r > 0$. We will need the following inequality due to Hardy and Littlewood. For every $f \in L_1([0, 1]^n)$ and $\epsilon > 0$, define the set 
\begin{equation*}
T(f, \epsilon) = \{x \, | \, \sup\limits_{r, t}\frac{\int_{I^t_r(x)} f d\mu}{\mu(I^t_r)} > \epsilon\},
\end{equation*}
where the supremum is taken over all $r > 0$ and $t \in \{-\frac{1}{3}, 0, \frac{1}{3}\}^n$. 
\begin{hardy}\label{hardy}
Let $f \in L_1([0, 1]^n)$ and $\epsilon > 0$. Then there exists a constant $c$ such that $\mu(T(f, \epsilon)) \leq \frac{c \|f\|_1}{\epsilon}$.
\end{hardy}

\begin{theorem}\label{integralrandom}
Let $f \in L_1([0, 1]^n)$ be polynomial space $L_1$ computable, approximated by the polynomial space computable sequence of step functions $\{f_m\}_{m \in \mathbb{N}}$. If $x$ is weakly pspace-random, then 
\begin{equation*}
\lim\limits_{m \rightarrow \infty}\frac{\int_{I^t_r(x)} |f-f_{m}|d\mu}{\mu(I^t_r(x))} = 0,
\end{equation*}
for every $t \in \{-\frac{1}{3}, 0, \frac{1}{3}\}^n$ and $r > 0$.
\end{theorem}
\begin{proof}
Let $p$ be a polynomial, and $f \in L_1([0, 1]^n)$ be polynomial space $L_1$ computable, approximated by the polynomial space computable sequence of simple step functions $\{f_m\}_{m \in \mathbb{N}}$ at precision $p$. For every $i > 0$, define the set 
\begin{equation*}
T_i = T(f_{2i-1} - f_{2i}, 2^{-i}) \cup T(f_{2i} - f_{2i + 1}, 2^{-i}). 
\end{equation*}
For every $m \geq 1$ define the set 
\begin{equation*}
U_m = \bigcup\limits_{i = m + 4 + c}^{\infty} T_i. 
\end{equation*}
We now prove that the sequence $\{U_m\}_{m \in \mathbb{N}}$ is a pspace $\mathcal{W}$-test. Clearly, for every $m > 0$, $U_m$ is an open set. By the Hardy/Littlewood inequality,
\begin{align*}
\mu(T_i) &\leq 2^i \, c \, (\|f_{2i - 1} - f_{2i} \| + \|f_{2i} - f_{2i + 1} \| )\\
&\leq 2^i \, c \, (2^{-2i + 2} + 2^{-2i + 1}) \\
&\leq c \, 2^{-i + 3}.
\end{align*}
Using this upper bound on the measure of $T_i$ we obtain
\begin{align*}
\mu(U_m) &\leq \sum\limits_{i = m + 4 + c}^\infty \mu(T_i) \\
&\leq \sum\limits_{i = m + 4 + c}^\infty c \, 2^{-i + 3} \\
&< 2^{-m}.
\end{align*}
It remains to be shown that the sequence $\{U_m\}_{m \in \mathbb{N}}$ is uniformly polynomial space approximable. By Lemma \ref{approxtounif}, it suffices to prove that the sequence $(T_i)$ is uniformly polynomial space approximable. 
For every $k$, $i$, define the sets 
\begin{equation*}
V^k_i = \{ I^t_r \, | \, r \leq p(2i + 1) + k + 2, \, t \in \{-\frac{1}{3}, 0, \frac{1}{3}\}^n,\mbox{ and } \, \frac{\int_{I^t_r(x)} |f_{2i - 1} - f_{2i}| d\mu}{\mu(I^t_r(x))} > 2^{-i}\}, 
\end{equation*}
\begin{equation*}
W^k_i = \{ I^t_r \, | \, r \leq p(2i + 1) + k + 2, \, t \in \{-\frac{1}{3}, 0, \frac{1}{3}\}^n,\mbox{ and } \, \frac{\int_{I^t_r(x)} |f_{2i} - f_{2i + 1}| d\mu}{\mu(I^t_r(x))} > 2^{-i}\},
\end{equation*}
and
\begin{equation*}
A^k_i = W^k_i \bigcup V^k_i. 
\end{equation*}

We now show that $\mu(T_i \Delta A^k_i) \leq 2^{-k}$. Intuitively, we bound the measure using the property that simple step functions are constant on dyadic cubes. Let $I^t_r \subseteq Q$, for some $Q \in \mathbf{B}_{p(2i + 1)}$; i.e., $I^t_r$ is fully contained in an open dyadic cube of precision $p(2i + 1)$. Assume
\begin{equation*}
\frac{\int_{I^t_r} |f_{2i - 1} - f_{2i}| d\mu}{\mu(I^t_r)} > 2^{-i}. 
\end{equation*}
Since $|f_{2i - 1} - f_{2i}|$ is a simple step function whose break points are in $\mathbf{D}^n_{p(2i + 1)}$, $|f_{2i - 1} - f_{2i}|$ must be a constant function on $Q$. Thus, $|f_{2i - 1}(Q) - f_{2i}(Q)| > 2^{-i}$, and so $I^t_r \subseteq Q \subseteq A^1_i$. Similarly, if 
\begin{equation*}
\frac{\int_{I^t_r} |f_{2i} - f_{2i + 1}| d\mu}{\mu(I^t_r)} > 2^{-i},
\end{equation*}
then $I^t_r \subseteq Q \subseteq A^1_i$. So then, the set of points in $T_i - A^k_i$ must be contained in some translate $I^t_r$ that is not contained in a dyadic cube of precision $p(2i+1)$; that is,
\begin{equation}\label{approxbound}
T_i - A^k_i \subseteq \bigcup\limits_{r = p(2i + 1) + k + 3}^\infty N_r.
\end{equation}

We now bound the measure of these points. For $r \in \mathbb{N}$ define the set
\begin{equation*}
N_r = \{ I^t_r \, | \, t \in \{-\frac{1}{3}, 0, \frac{1}{3}\}^n,\mbox{ and }I^t_r \nsubseteq Q\mbox{ for any box } Q \text{ of precision }p(2i + 1)\}.
\end{equation*}
If $I^t_r$ is not contained in a dyadic cube of precision $p(2i + 1)$, then $I^t_r$ must contain at least one dyadic rational of precision $p(2i + 1)$. Hence,
\begin{equation}
|N_r| \leq 3^n 2^{n p(2i + 1)}
\end{equation}
and so,
\begin{equation}\label{measurNr}
\mu(N_r) \leq 3^n 2^{n p(2i + 1)} \, 2^{-rn}.
\end{equation}

By equation \ref{approxbound} and inequality \ref{measurNr}, we obtain
\begin{align*}
\mu(T_i - A^k_i) &\leq \mu( \bigcup\limits_{r = p(2i + 1) + k + 3}^\infty N_r ) \\
				 &\leq \sum\limits_{r = p(2i + 1) + k + 3}^\infty \mu(N_r) \\
				 &\leq \sum\limits_{r = p(2i + 1) + k + 3}^\infty 3^n \, 2^{np(2i + 1)} \, 2^{-rn} \\
				 &\leq 3^n \,2^{np(2i + 1)}  \sum\limits_{r = p(2i + 1) + k + 3}^\infty 2^{-rn} \\
				 &\leq 2^{-k - 1}.
\end{align*}

We would like $\{A^k_i\}$ to be a uniformly polynomial space computable sequence. However, there is a minor technical detail which complicates the argument. The definition of uniformly pspace computable sequences requires the endpoints to be dyadic rationals. Unfortunately, translating the dyadic cubes by $t \in \{-\frac{1}{3}, 0, \frac{1}{3}\}^n$ violates this requirement. In order to overcome this, we will approximate $\{A^k_i\}$ by boxes with dyadic endpoints. For any open cube $Q$, define $D^k_i(Q)$ to be the open dyadic box containing $Q$ such that 
\begin{equation*}
\mu(D^k_i(Q) - Q) < 2^{-n(p(2i + 1) + 2k + 3)}. 
\end{equation*}
Formally, if $Q = (a_1, b_1) \times \ldots \times (a_n, b_n)$, let
\begin{equation*}
D^k_i(Q) = (d_1, d^\prime_1) \times \ldots \times (d_n, d^\prime_n)
\end{equation*}
where $d_i$, $d^\prime_i$ are dyadic rationals at precision $p(2i + 1) + 2k + n + 3$, and $d_i \leq a_i < b_i \leq d^\prime_i$.

Define the set 
\begin{equation*}
S^k_i = \bigcup\limits_{Q \in A^k_i} D^k_i(Q).
\end{equation*}
It is easy to verify that $\{S^k_i\}$ is a uniformly pspace computable array such that the endpoints of $S^k_i$ are in $\mathbf{D}^n_{p(2i + 1) + 2k + n + 3}$, and $\mu(T_i \Delta S^k_i) \leq 2^{-k}$ for every $i$, $k > 0$. It is clear that, for every $i$ and all $x \in T_i$, $x \in \liminf_k S^k_i$. Hence, $\{T_i\}$ is a uniformly polynomial space approximable sequence, and $\{U_m\}_{m \in \mathbb{N}}$ is a pspace $\mathcal{W}$-test.

Assume $x$ is weakly pspace-random. Then there exists an $N$ such that for all $m > N$, $x \notin U_m$. Let $i > 2N + 8 + 2c$, $t \in \{-\frac{1}{3}, 0, \frac{1}{3}\}^n$ and $r > 0$. Choose $j > rn + i$. Then,
\begin{align*}
\frac{\int_{I^t_r(x)} |f-f_{i}|d\mu}{\mu(I^t_r(x))} &\leq \frac{\int_{I^t_r(x)} |f- f_{j}|d\mu}{\mu(I^t_r(x))} + \frac{\int_{I^t_r(x)} |f_j- f_{i}|d\mu}{\mu(I^t_r(x))} \\
&\leq 2^{rn} \, 2^{-j} + \frac{\int_{I^t_r(x)} |f_j- f_{i}|d\mu}{\mu(I^t_r(x))} \\
&\leq 2^{-i} + \sum\limits_{m = i}^{j - 1}\frac{\int_{I^t_r(x)} |f_m- f_{m + 1}|d\mu}{\mu(I^t_r(x))} \\
&\leq 2^{-i} + \sum\limits_{m = i}^{j - 1} 2^{-\frac{m}{2}} \\
&\leq 2^{-i} + 2^{-\frac{i}{2} + 2} \\
&< 2^{-\frac{i}{2} + 3}.
\end{align*}
Since $t \in \{-\frac{1}{3}, 0, \frac{1}{3}\}^n$ and $r > 0$ were arbitrary, 
\begin{equation*}
\lim\limits_{m \rightarrow \infty}\frac{\int_{I^t_r(x)} |f-f_{m}|d\mu}{\mu(I^t_r(x))} = 0,
\end{equation*}
for every $t \in \{-\frac{1}{3}, 0, \frac{1}{3}\}^n$ and $r > 0$.
\end{proof}
We are now able to prove that weakly pspace random points satisfy the Lebesgue differentiation theorem.
\begin{theorem}\label{theorem:WeaklySatisfyLebesgue}
If $x$ is weakly pspace-random, then for every polynomial space $L_1$-computable $f \in L_1([0, 1]^n)$, and every polynomial space computable sequence of simple functions $\{f_m\}_{m \in \mathbb{N}}$ approximating $f$,
\begin{equation*}
\lim_{m \rightarrow \infty} f_m(x) = \lim_{Q \rightarrow x} \frac{\int_Q f d\mu}{\mu(Q)}
\end{equation*}
where the limit is taken over all cubes $Q$ containing $x$ as the diameter of $Q$ tends to 0.
\end{theorem}
\begin{proof}
Let $x$ be weakly pspace-random. By Theorem \ref{rutestheorem}, it suffices to show that 
\begin{equation*}
\lim_{m \rightarrow \infty} f_m(x) = \lim_{k \rightarrow \infty} \frac{\int_{I^t_k(x)} f d\mu}{\mu(I^t_k(x))}
\end{equation*}
for all $t \in \{-\frac{1}{3}, 0, \frac{1}{3}\}^n$.

Let $\epsilon > 0$. By Theorems \ref{limitrandom} and \ref{integralrandom}, there exists an $N$ such that for all $i > N$, 
\begin{equation}\label{meq1}
|f_{i}(x) - \lim_{m \rightarrow \infty} f_m(x)| < \frac{\epsilon}{2},
\end{equation}
and
\begin{equation}\label{meq2}
\frac{\int_{I^t_k(x)} |f-f_{i}|d\mu}{\mu(I^t_k(x))} < \frac{\epsilon}{2},
\end{equation}
for every $t \in \{-\frac{1}{3}, 0, \frac{1}{3}\}^n$ and $k > 0$.
Let $i > N$. Then, using (\ref{meq1}) we obtain
\begin{equation}\label{mainlim}
|\lim_{m \rightarrow \infty} f_m(x) - \lim_{k \rightarrow \infty} \frac{\int_{I^t_k(x)} f d\mu}{\mu(I^t_k(x))}| < \frac{\epsilon}{2} + |f_i(x) - \lim_{k \rightarrow \infty} \frac{\int_{I^t_k(x)} f d\mu}{\mu(I^t_k(x))}|.
\end{equation}
By Lemma \ref{dyadic2random}, for every $r > 0$, $x \in Q$ for some $Q \in \mathbf{B}_{r}$. Since $f_i$ is a simple step function, $f_i$ is constant on every $Q \in \mathbf{B}_{p(i)}$. So there exists an $N^\prime$ so that for all $r > N^\prime$, 
\begin{equation*}\label{meq3}
f_i(x) = \frac{\int_{I^t_r(x)} f_i d\mu}{\mu(I^t_r(x))},
\end{equation*}
for every $t \in \{-\frac{1}{3}, 0, \frac{1}{3}\}^n$. Therefore, by inequality \ref{meq2}, for every $r > N^\prime$,
\begin{align}
|f_i(x) - \frac{\int_{I^t_r(x)} f d\mu}{\mu(I^t_r(x))}| &= |\frac{\int_{I^t_r(x)} f_i d\mu}{\mu(I^t_r(x))} - \frac{\int_{I^t_r(x)} f d\mu}{\mu(I^t_r(x))}|\\ 
&\leq \frac{\int_{I^t_r(x)} |f - f_i| d\mu}{\mu(I^t_r(x))}\\
&< \frac{\epsilon}{2}\label{mainintegral}.
\end{align}
Combining inequalities (\ref{mainlim}) and (\ref{mainintegral}) we have
\begin{equation*}
|\lim_{m \rightarrow \infty} f_m(x) - \lim_{k \rightarrow \infty} \frac{\int_{I^t_k(x)} f d\mu}{\mu(I^t_k(x))}| < \epsilon.
\end{equation*}
Since $\epsilon$ was arbitrary, the proof is complete.
\end{proof}

\subsection{Non-random points are not Lebesgue points}
We now show that converse of our main theorem holds. That is, we show that if a point $x$ is not weakly pspace random, the limit $\lim\limits_{Q \rightarrow x} \frac{1}{\mu(Q)}\int_Q fd\mu$ does not exist. Our approach is largely similar from the construction of Pathak, et al \cite{PRS}. However, due to the restriction of polynomial space computation, the implementation is significantly different. To adapt the construction of Pathak et al, we first introduce a notion that will partition a pspace $\mathcal{W}$-test $\{U_m\}$ into a tree of dyadic cubes. 

\begin{definition}\label{dyadicTreeDef}
A \textit{dyadic tree decomposition of} $[0, 1]^n$ is a tree $\mathbf{T}$ of dyadic cubes rooted at $[0, 1]^n$ such that the following hold:
\begin{enumerate}
\item The children of any cube $Q \in \mathbf{T}$ are subsets of $Q$.
\item For any two cubes $Q_1, Q_2 \in \mathbf{T}$, either $Q_1$ and $Q_2$ are disjoint, or one contains the other.
\item For any cube $Q \in \mathbf{T}$, 
\begin{equation*}
\mu(\bigcup\limits_{B \in Child(Q)} V) < \frac{\mu(Q)}{4}.
\end{equation*}
\end{enumerate}

A dyadic tree decomposition $\mathbf{T}$ is \textit{polynomial space approximable} if there exists a polynomial $p$ and uniformly pspace computable array $\{T^k_m\}_{k,m \in \mathbb{N}}$ such that the following hold.
\begin{enumerate}
\item For every $k, m \in \mathbb{N}$, $T^k_m$ is a finite union of disjoint dyadic cubes.
\item For every $m \in \mathbb{N}$, $T^k_m$ approximates the $m$th level of $T$ to within $2^{-k}$, that is, $\mu(Level_{m}(\mathbf{T}) \Delta T^k_m) < 2^{-k}$.
\end{enumerate}
\end{definition}

The following technical lemma will be used to show that every pspace $\mathcal{W}$-test admits a pspace approximable dyadic tree decomposition.
\begin{lemma}\label{lemma:NiceArray}
For every uniformly pspace computable array $\{R^k_m\}_{k, m \in \mathbb{N}}$, there exists a uniformly pspace computable array $\{S^k_m\}_{k, m \in \mathbb{N}}$ such that 
\begin{enumerate}
\item For every $m, k$, $\mu(\cup_{i \leq k} R^i_m \Delta \cup_{i\leq k} S^i_m) = 0$, and
\item For every $m$, $\cup_k S^k_m$ is a set of disjoint open dyadic cubes.
\end{enumerate}
\end{lemma}
\begin{proof}
We can, and do assume that, for every $k, m$, $R^k_m$ is a finite union of disjoint open dyadic cubes, whose endpoints are dyadic rationals at precision $p(k + m)$. For every $m$, define $S^1_m = R^1_m$. Let $m \in \mathbb{N}$ and $k > 1$. Define the set 

\begin{equation*}
A^k_m = \{Q \in \mathbf{B}_{p(k + m)} \, | \, (\exists i < k) \, Q \subseteq B \text{ where } B \in R^i_m\}.
\end{equation*}
That is, $A^k_m$ is the set of all cubes in $\cup_{i < k} R^i_m$ broken into dyadic cubes of precision $p(k+m)$. Define $S^k_m = R^k_m - A^k_m$. 

It is clear that $\{S^k_m\}_{k, m \in \mathbb{N}}$ satisfies both properties of the lemma. Note that $\{A^k_m\}_{m \in \mathbb{N}, k> 1}$ is a uniformly pspace computable array. It therefore follows that $\{S^k_m\}_{k, m \in \mathbb{N}}$ is pspace computable.
\end{proof}

We now show that every pspace $\mathcal{W}$-test admits a pspace approximable dyadic tree decomposition. We build the tree inductively, using the uniformly pspace computable sequence of the previous lemma. 
\begin{lemma}\label{lemma:WtestHasDyadicTree}
Let $\{U_m\}_{m\in \mathbb{N}}$ be a pspace $\mathcal{W}$-test. Then there exists a pspace approximable dyadic tree decomposition $\mathbf{T}$ such that, for every non-dyadic $x \in \bigcap U_m$, $x$ is contained in an infinite path in $\mathbf{T}$.
\end{lemma}
\begin{proof}
Let $\{U_m\}_{m\in \mathbb{N}}$ be a pspace $\mathcal{W}$-test. Let $\{R^k_m\}_{k, m \in \mathbb{N}}$ be a uniformly pspace computable array approximating $\{U_m\}_{m\in \mathbb{N}}$. We can and do assume that for all $k, m \in \mathbb{N}$, $\mu(U_m \Delta R^k_m) < 2^{-(k + m)}$. Let $\{S^k_m\}_{k, m \in \mathbb{N}}$ be the uniformly pspace computable array of obtained from $\{R^k_m\}$ satisfying the properties of Lemma \ref{lemma:NiceArray}. For every $m$, define the set
\begin{equation*}
S_m = \{Q \, | \, Q \in S^k_m \text{ for some } k \geq 1\}.
\end{equation*}

We define the dyadic tree decomposition $\mathbf{T}$ inductively. Define the first level of $\mathbf{T}$ to be
\begin{equation*}
Level_1(\mathbf{T}) = S_1.
\end{equation*}
For $i > 1$, define level $i$ as follows. For every $Q \in Level_{i - 1}(\mathbf{T})$, let $m \in \mathbb{N}$ be the smallest integer such that $2^{-m} < \frac{\mu(Q)}{8}$. Define the set
\begin{equation*}
Child(Q) = \{ B \, | \, B \in S_m \text{ and } B \subseteq Q\}.
\end{equation*}
Finally, define the $i$th level to be
\begin{equation*}
Level_i(\mathbf{T}) = \bigcup\limits_{Q \in Level_{i - 1}(\mathbf{T})} Child(Q).
\end{equation*}

We now prove that $\mathbf{T}$ is a dyadic tree decomposition of $[0, 1]^n$. By our construction of $\mathbf{T}$, it is clear that for any $Q \in \mathbf{T}$, the children of $Q$ are subsets of $Q$. We prove item (2) of definition \ref{dyadicTreeDef} by induction on the level of the tree. For the root $[0,1]^n$, the claim is immediate. Let $i > 0$. Let $Q_1, Q_2$ be dyadic cubes at level $i$. If $Q_1$ and $Q_2$ have different parents, the claim holds by our inductive hypothesis. Assume that $Q_1$ and $Q_2$ have the same parent. Then $Q_1, Q_2 \in \cup_k S^k_m$ for some $m \in \mathbb{N}$, and therefore $Q_1$ and $Q_2$ are disjoint. Let $Q \in \mathbf{T}$ and $m$ be the smallest integer such that $2^{-m} < \frac{\mu(Q)}{8}$. By the construction of $\mathbf{T}$, 
\begin{align*}
\mu(\bigcup\limits_{B \in Child(Q)} B) &\leq \mu(\cup_{k\geq 1} S^k_m)\\
&\leq \sum\limits^\infty_{k = 1} \mu(R^k_m)\\
&\leq \mu(U_m) + \mu(\bigcup_{k \geq 1} U_m \Delta R^k_m) \\
&\leq 2^{-m} + \sum\limits^\infty_{k = 1} 2^{-(k + m)} \\
&= 2^{-m + 1} \\
&\leq \frac{\mu(Q)}{4}
\end{align*}

We now show that $\mathbf{T}$ is pspace approximable. We define the array $\{T^k_m\}_{k, m \in \mathbb{N}}$ inductively on $m$. For $m = 1$, set 
\begin{equation*}
T^k_1 = \bigcup\limits_{i = 1}^k S^i_1.
\end{equation*}
Let $m > 1$ and $k \in \mathbb{N}$. For every $Q \in T^k_{m - 1}$, let $j \in \mathbb{N}$ be the smallest integer such that $2^{-j} < \frac{\mu(Q)}{8}$. Define the set
\begin{equation*}
C^k_Q = \{ B \in Child(Q)\, | \, B \in S^i_j \text{ for some } i \leq k + 3\}.
\end{equation*} 
Since, 
\begin{align*}
\mu(\bigcup\limits_{i = k + 3}^\infty S^i_j) &\leq \sum\limits_{i = k + 3}^\infty \mu(S^i_j)\\
&\leq \sum\limits_{i = k + 3}^\infty \mu(R^i_j \Delta R^{i - 1}_j)\\
&\leq \sum\limits_{i = k + 3}^\infty \mu(U_j \Delta R^i_j) + \mu(U_j \Delta R^{i - 1}_j)\\
&\leq \sum\limits_{i = k + 3}^\infty 2^{-(j + i)} + 2^{-(j + i - 1)}\\
&\leq 2^{-(j + k)},
\end{align*}
we have
\begin{align*}
\mu(Child(Q) - C^k_Q) &\leq 2^{-(j + k + 2)}\\ 
&\leq \frac{\mu(Q)}{8} 2^{-k}.
\end{align*}
Finally, define 
\begin{equation*}
T^k_m = \bigcup\limits_{Q \in T^k_{m - 1}} C^k_Q.
\end{equation*}

We now show that $\{T^k_m\}_{k, m \in \mathbb{N}}$ approximates $\mathbf{T}$ by induction on the level $m$. It is clear that for all $k$, $\mu(Level_1(\mathbf{T}) \Delta T^k_1) < 2^{-k}$. Let $k, m \in \mathbb{N}$. Define the set 
\begin{equation*}
N = \{Q \, | \, Q \in Level_{m-1}(\mathbf{T}) - T^k_{m-1}\}. 
\end{equation*}
Then,
\begin{align*}
\mu(Level_m(\mathbf{T}) \Delta T^k_m) &= \mu(\bigcup\limits_{Q \in T^k_{m-1}} Child(Q) \Delta T^k_m) + \mu(\bigcup\limits_{Q \in N} Child(Q))\\
&\leq \sum\limits_{Q \in T^k_{m-1}}\mu(Child(Q) - C^k_Q) + \sum\limits_{Q \in N}\mu( Child(Q))\\
&\leq \sum\limits_{Q \in T^k_{m-1}}(\frac{\mu(Q)}{8} 2^{-k }) + 2^{-(k + 3)}\\
&\leq 2^{-k}.
\end{align*}
Since $\{S^k_m\}_{k, m \in \mathbb{N}}$ is pspace computable, $\{T^k_m\}_{k, m \in \mathbb{N}}$ is a uniformly pspace computable array. Hence $\mathbf{T}$ is a pspace approximable dyadic tree decomposition.

Let $x = (x_1, \ldots, x_n) \in \cap_{m \geq 1} U_m$ be a point so that $x_i$ is not a dyadic rational. We prove that there is an infinite path in $\mathbf{T}$ containing $x$ by induction on the level of $\mathbf{T}$. By the definition of pspace $\mathcal{W}$-tests, it is clear that there exists a dyadic cube $Q$ in $S_1$ such that $x \in Q$. Hence $Q \in Level_1(\mathbf{T})$. Let $i > 1$. By our inductive hypothesis, there exists a dyadic rational cube $Q \in Level_{i - 1}(\mathbf{T})$ containing $x$. Let $m$ be the smallest integer such that $2^{-m} < \frac{\mu(Q)}{8}$. Since there exists a dyadic cube $Q \in S_m$ containing $x$, the conclusion follows. 
\end{proof}

We are now able to prove the converse of Theorem \ref{theorem:WeaklySatisfyLebesgue}, thereby completing the proof of our main theorem. The proof of this theorem involves constructing a function that takes advantage of the dyadic tree decomposition of a pspace $W$-test succeeding on $x$. We construct the function so that it assigns different values to alternating levels of the tree. As we are guaranteed that $x$ is in an infinite path of the tree, the function oscillates around $x$.
\begin{theorem}\label{theorem:MainThmConverse}
If $x \in [0, 1]^n$ is not weakly pspace random, then there exists a pspace $L_1$ computable function $f$ such that the limit $\lim\limits_{Q \rightarrow x} \frac{1}{\mu(Q)}\int_Q fd\mu$ does not exist.
\end{theorem}
\begin{proof}
We first assume that $x = (x_1, \ldots, x_n)$ so that some component $x_i$ of $x$ is a dyadic rational. Without loss of generality assume that $x_1 = d \in \mathbf{D}$. Define the function $f: [0 ,1]^n \rightarrow \mathbb{R}$ to be
\begin{equation*}
f(y) = \begin{cases}
1 &\text{ if } y \in [0, d] \times [0, 1] \times \ldots \times [0,1]\\
0 &\text{ otherwise}
\end{cases}
\end{equation*}
It is clear that $f$ is pspace $L_1$-computable, and that the limit $\lim\limits_{Q \rightarrow x} \frac{1}{\mu(Q)}\int_Q fd\mu$ does not exist.

Assume that $x = (x_1, \ldots, x_n)$ so that $x_i$ is not a dyadic rational for all $i \leq n$. Let $\{U_m\}_{m\in \mathbb{N}}$ be a pspace $\mathcal{W}$-test succeeding on $x$. Let $\mathbf{T}$ be a pspace computable dyadic tree partition of $\{U_m\}_{m\in \mathbb{N}}$ given by Lemma \ref{lemma:WtestHasDyadicTree}. Define $f:[0, 1]^n \rightarrow \mathbb{R}$ as follows. For every $Q \in \mathbf{T}$,
\begin{align*}
f(Q - \bigcup\limits_{B \in Child(Q)} B) = \begin{cases}
1 &\text{ if the level of } Q \text{ in } \mathbf{T} \text{ is even}\\
0 &\text{ if the level of } Q \text{ in } \mathbf{T} \text{ is odd}
\end{cases}
\end{align*}

We now show that $f$ is pspace $L_1$-computable. Let $\{T^k_m\}_{k, m \in \mathbb{N}}$ be the uniformly pspace computable array approximating $\mathbf{T}$. For every $m \in \mathbb{N}$, define
\begin{equation*}
T_m = \{Q \, | \, Q \in T^{m+2}_i \text{ for some } i \leq m\}. 
\end{equation*}
For every $Q \in T_m$, let $C^k_Q = \{B \in Child(Q) \, | \, B \in T_{m}\}$. For every $m \in \mathbb{N}$, define $f_m:[0, 1]^n \rightarrow \mathbb{R}$ as follows.
\begin{align*}
f_m(Q - \bigcup\limits_{B \in C^k_Q} B) = \begin{cases}
1 &\text{ if the level of } Q \text{ in } \mathbf{T} \text{ is even}\\
0 &\text{ if the level of } Q \text{ in } \mathbf{T} \text{ is odd}
\end{cases}
\end{align*}
It is clear that $f_m$ is a simple step function. On input $(0^m, d)$, the machine computes the sequence of descending dyadic cubes in $T_m$ containing $d$, and outputs appropriately. Since $\{T^k_m\}$ is pspace computable, computing this sequence can be done in polynomial space, and $\{f_m\}$ is pspace computable.

We now prove that $\{f_m\}_{m\in \mathbb{N}}$ approximates $f$. Let $m \in \mathbb{N}$, and
\begin{equation*} 
N = \{Q \, | \, Q \in \mathbf{T} - \cup_{i = 1}^m T^{m+2}_i\}. 
\end{equation*}
Then,
\begin{align*}
\|f - f_m\|_1 &= \int_0^1 |f - f_m| \\
&= \int_Q |f - f_m| \\
&\leq \mu(N)\\
&\leq \sum\limits_{i = 1}^m \mu(Level_i(\mathbf{T}) - T^{m+1}_i) + \sum\limits_{i = m + 1}^\infty \mu(Level_i(\mathbf{T}))\\
&\leq 2^{-(m + 1)} + 2^{-(m + 1)}\\
&\leq 2^{-m}.
\end{align*}
Hence $f$ is a pspace $L_1$ computable function.

Finally, we show that the limit $\lim\limits_{Q \rightarrow x} \frac{1}{\mu(Q)}\int_Q fd\mu$ does not exist. Let $N \in \mathbb{N}$. By Lemma \ref{lemma:WtestHasDyadicTree}, $x$ is contained in an infinite path of $\mathbf{T}$. Choose a dyadic cube $Q \in \mathbf{T}$ containing $x$ such that $\mu(Q) < 2^{-N}$ and the level of $Q$ in $\mathbf{T}$ is even. Then, by our construction of $f$,
\begin{align*}
\frac{1}{\mu(Q)}\int_Q fd\mu &\geq \frac{1}{\mu(Q)}\int_{Q - Child(Q)} 1d\mu\\
&= \frac{1}{\mu(Q)}\mu(Q - Child(Q))\\
&\geq \frac{3}{4}.
\end{align*}
Similarly, choose a dyadic cube $Q \in \mathbf{T}$ containing $x$ such that $\mu(Q) < 2^{-N}$ and the level of $Q$ in $\mathbf{T}$ is odd. Then, by construction of $f$,
\begin{align*}
\frac{1}{\mu(Q)}\int_Q fd\mu &\leq \frac{1}{\mu(Q)}\int_{Child(Q)} 1d\mu\\
&= \frac{1}{\mu(Q)}\mu(Child(Q))\\
&\leq \frac{1}{4}.
\end{align*}
Hence the limit $\lim\limits_{Q \rightarrow x} \frac{1}{\mu(Q)}\int_Q fd\mu$ does not exist.
\end{proof}

Therefore, by Theorems \ref{theorem:WeaklySatisfyLebesgue} and \ref{theorem:MainThmConverse}, the Lebesgue differentiation theorem characterizes weakly pspace randomness.

\section{Conclusion and Open Problems}
In the computable setting, there is a strong connection between randomness and classical theorems of analysis. However, this interaction is not as well understood in the context of resource-bounded randomness. An interesting direction is to characterize randomness for different computational resource bounds using the Lebesgue differentiation theorem. For example, what notion of polynomial time randomness is characterized by the Lebesgue differentiation theorem?

We believe the notion of weakly polynomial space randomness will be useful in further investigations into resource-bounded randomness in analysis. An interesting avenue of future research is to relate weakly pspace-randomness with other notions of polynomial space randomness. We showed that Lutz's definition of pspace-randomness implies weakly pspace randomness, but the converse is not known. We conjecture that weakly pspace randomness is strictly weaker than Lutz's notion of pspace-randomness.

\subparagraph{Acknowledgments}
We thank Adam Case and Jack Lutz for useful discussions.

\end{document}